\documentclass[aps,pra,twocolumn,showkeys]{revtex4-1}
\usepackage{amsmath,paralist,amsthm,comment,amssymb}
\usepackage{tikz}

\DeclareMathOperator{\Tr}{Tr}
\DeclareMathOperator{\rk}{rank}
\DeclareMathOperator{\im}{im}
\DeclareMathOperator{\s}{{\mathsf{S}}}

\newcommand{\mc}[1]{\mathcal{#1}}
\newcommand{\mbf}[1]{\mathbf{#1}}

\newcommand{\F}{\mathbb{F}}
\newcommand{\nix}[1]{}
\newcommand{\ket}[1]{|#1\rangle}

\newcommand{\be}{\begin{eqnarray*}}
\newcommand{\ee}{\end{eqnarray*}}
\newcommand{\ben}{\begin{eqnarray}}
\newcommand{\een}{\end{eqnarray}}

\newcommand{\ba}{\begin{array}}
\newcommand{\ea}{\end{array}}

\newtheorem{theorem}{Theorem}
\newtheorem{corollary}[theorem]{Corollary}
\newtheorem{lemma}{Lemma}

\newtheorem{defn}{Definition}

\begin{document}
\title{ Entropic Inequalities for a Class of Quantum Secret  Sharing States}
\author{Pradeep Sarvepalli}
\email[]{pradeep@phas.ubc.ca}
\affiliation{Department of Physics and Astronomy, University of British Columbia, Vancouver, BC, V6T 1Z1}

\date{September 2, 2010}

\begin{abstract}
It is well-known that von Neumann entropy is nonmonotonic  unlike Shannon entropy 
(which is monotonically nondecreasing). Consequently, it is difficult to relate the 
entropies of the subsystems  of a given quantum state. In this paper, we  show that if 
we consider  quantum secret sharing states arising from a class of monotone span programs,  
then we can partially recover the monotonicity of entropy for the so-called unauthorized 
sets. Furthermore, we can show for these quantum states the entropy of the authorized sets 
is monotonically nonincreasing.
\end{abstract}


\maketitle
\section{Introduction}

In this paper we   are motivated by the fact that one of the reasons why von Neumann 
entropy  behaves differently from the Shannon entropy is rooted in its 
nonmonotonicity.  A consequence of this fact is that many results and techniques of 
classical information theory do not smoothly generalize; one has to frequently overcome 
the obstacles imposed by the breakdown of monotonicity, often manifesting in various 
disguises. In classical secret sharing, one makes extensive
use of information theoretic inequalities to bound the performance of secret sharing 
schemes, see  \cite{capocelli93,kurosawa94,blundo94,vandijk95,csirmaz97,beimel08}. In most cases these techniques do not appear to carry over to quantum secret 
sharing schemes, at least not easily. Quantum secret sharing has grown rapidly since its inception in \cite{hillery99}. Despite the growing body of literature on quantum secret sharing, see for instance \cite{hillery99, cleve99, gottesman00, smith00, markham08,keet10} and the
references therein, only a few of them,  most notably \cite{imai04,rietjens05}, have succeeded in employing information theoretic methods for quantum secret sharing.

This paper attempts to  contribute along this direction by studying the von Neumann entropy of subsets of a given quantum 
state. Of course, for an arbitrary quantum state we cannot expect a relation similar to the
Shannon entropy. However, by imposing some restrictions on the  quantum states, namely
by  considering a class of quantum secret sharing states, 
we are able to prove something definite. 

Our main result is that if we consider 
a class of secret sharing quantum states, then we can prove that the 
von Neumann entropy is
monotonically nondecreasing for certain subsystems of the quantum state, namely the 
unauthorized sets. The secret sharing states that we consider in this correspondence are those
arising from a realization of the quantum secret sharing scheme via monotone span programs in a ``normal form''. 
We also show that for the subsystems that are authorized, the 
entropy is monotonically nonincreasing.

\section{Background}
Before presenting our main result, we review briefly the relevant background. In this 
paper a quantum secret sharing scheme refers to a protocol to distribute an unknown 
quantum state to a group of $n$ players so that only authorized subsets of players can 
reconstruct the secret quantum state \cite{hillery99, cleve99, gottesman00}. The state 
received by any participant is called a share. 
The set of players is denoted as $P= \{1,2,\ldots,n \}$. The 
collection of authorized sets, denoted as $\Gamma$, is called the access structure
and the collection of unauthorized sets, denoted $\mc{A}$, is called the adversary 
structure. The dual access structure is defined as 
\ben
\Gamma^\ast = \{ A\subseteq P \mid \bar{A} \not\in \Gamma \},\label{eq:dualAcc}
\een
where $\bar{A}= P\setminus A$. 

An access structure $\Gamma$ can be realized by a  quantum
secret sharing scheme if and only if it satisfies $\Gamma \subseteq \Gamma^\ast$, see \cite{smith00,gottesman00}. 
An access structure is said to be self-dual if $\Gamma=\Gamma^\ast$. A self-dual access structure can be realized as quantum secret sharing scheme so that pure states are encoded as pure states; the associated 
secret sharing scheme is said to be a pure-state scheme. If $\Gamma\subsetneq \Gamma^\ast$, then the associated secret sharing scheme encodes some pure state into a mixed state;
such schemes are called mixed-state schemes.

An authorized set is said to be minimal if every proper subset of it is unauthorized. 
Minimal authorized sets completely characterize an access structure and  we denote 
the collection of minimal authorized sets of $\Gamma$ by $\Gamma_{\min}$. 

A secret sharing scheme is said to be connected if every participant occurs in some 
minimal authorized set. 
A participant who does not occur in any minimal authorized set is said to be unimportant
and the share associated to that player can be discarded without loss. 
We can assume that the access structure is defined on the reduced  set
of players excluding the unimportant players. In this paper we
consider only connected secret sharing schemes, in other words, every player is important.

Closely related to this idea of important share is the notion of dispensable components
of a share \footnote{I would like to thank Daniel Gottesman for bringing up this point.}. Suppose a part of the share received by a player does not play a role in the
reconstruction of the secret in any of the minimal authorized sets to which the player
belongs, then that part of the share can be discarded without any loss. We call such parts
of the share dispensable components and indispensable otherwise. One way to create such
components is to distribute additional quantum states that are unrelated to the secret.
Such components neither aid nor hinder a subset's 
ability to reconstruct the secret.
The  access structure realized by the scheme is unaffected if such 
components of the share are discarded. We assume that the secret sharing schemes considered in this paper are devoid of such dispensable components in any share. 

\begin{defn}[CSS secret sharing schemes]
A pure-state secret sharing scheme  $\Sigma$ is said to be a CSS secret sharing scheme if 
the encoding for the scheme is of the form 
\ben
\ket{s} \mapsto \frac{1}{\sqrt{|C|}}\sum_{c\in C} \ket{s \overline{X}+c},\label{eq:cssScheme}
\een
where $s\in \F_q$ and $\overline{X}\in \F_q^n$. 
A mixed-state secret scheme is said to be CSS type if it can be obtained by discarding one
share of a pure-state CSS scheme. 
\end{defn}
As mixed-state schemes can be obtained from pure-state schemes, see \cite{gottesman00},   without loss of generality we can focus our attention on pure-state schemes. In terms of access structures this implies we can focus on the self-dual access structures. 

Smith proposed a method to realize a quantum secret sharing scheme for any quantum
access structure using monotone span programs. We will need parts of that construction
to prove our result and therefore provide a quick review of the same. Further details can
be found in \cite{smith00}. 

\begin{defn}[Montone span program]
A monotone span program $\mc{M}$ is a triple $(\F_q, M, \psi)$ consisting of a
matrix $M\in \F_q^{d\times e}$ over $\F_q$ and surjective function $\psi : \{1,2,\ldots, d \} \rightarrow P$. The program is said to accept $A\subseteq P$ if and only if 
\ben
\varepsilon_1=(1,0,\ldots, 0)^t \in \im(M_A^t), \label{eq:mspCond}
\een
where $M_A$ is the submatrix of $M$ and consists of all the rows that correspond to the participants in $A$.
\end{defn}

We say $M$ is the span matrix of the span program. 
A monotone span program is said to compute an access structure $\Gamma$  if 
$\mc{M}$ accepts $A\in \Gamma$ and rejects $A\not\in \Gamma$.
It follows that for any unauthorized
set $B\subseteq P$, we must have $\varepsilon_1 \not\in \text{Im}(M_B^t)$. Alternatively,
there exists some $v=(v_1,\ldots, v_e )\in \F_q^e$ such that $M_B v = 0 $ and $v_1\neq 0$. 

For every access structure $\Gamma$, we can realize it using a 
monotone span program in a ``normal form''. This is not necessarily efficient 
but these realizations provide us with the secret sharing states which are central for our 
discussion.  We caution the reader that our terminology of  span programs, especially the 
usage of normal form for  monotone span programs,  does not follow the standard 
terminology in the literature on span programs.

\begin{defn}[Normal form monotone span program]
Suppose that $\Gamma_{\min} =\{ A_1,\ldots A_k\}$ and let $|A_i|-1=r_i$.
Also let, $r=\sum_{i} |A_i|$ and $c=r-k=\sum_i r_i$. The monotone span program that
computes $\Gamma_{\min}$ over $\F_q$ has the span matrix $M$ as given in equation~\eqref{eq:msp}.  The mapping 
$\psi$ is given as follows: 
The rows corresponding to submatrix $I_{r_i}$ corresponds
to the first $|A_i|-1$ participants of the authorized set $A_i$. The row of the form 
$(1, 0,\ldots, 0,-\mbf{1},0, \ldots, 0)$ is labeled by the last participant of  $A_i$ and
$-\mbf{1} =(-1,\ldots,-1) \in \F_q^{r_i}$.
\end{defn}
\ben
M= \left[ \ba{ccccccc}
&I_{r_1}&&&&\\
1&-\mbf{1}&&\\ \hline
&&I_{r_2} &&&\\
1&&-\mbf{1}&\\ \hline
\vdots&&&\ddots\\ \hline
&&& &I_{r_i} \\
1&&&&-\mbf{1}\\ \hline
\vdots&&&&&\ddots\\ \hline
&&&&& &I_{r_k} \\
1&&&&&&-\mbf{1}\\
\ea\right]\label{eq:msp}
\een

Classically this matrix defines the encoding for the secret sharing scheme. 
The scheme implicitly assumes an ordering of the participants and the use
of $c$ random variables in addition to the secret. For each minimal authorized set $A_i$, the scheme distributes
$s_{i,1},s_{i,2}\ldots, s_{i,r_i}$ and $s-\sum_j s_{i,j}$, where 
$s_{i,j}$ for $1\leq j\leq r_i$ are random variables. More precisely, 
$s_{i,j}$ is distributed to the $j$th player of $A_i$ for $1\leq j\leq r_i$
and $s-\sum_{j}s_{i,j}$ to the last player in $A_i$. 
Considering the distribution scheme for the entire set of players we distribute 
 $u M^t$, 
where $u =(s, s_1,\ldots, s_{c})$. 
In the normal form construction, no player is unimportant and no share has
a dispensable component.

We make the following useful observations about the structure of $M$ as given in equation~\eqref{eq:msp}.
\begin{compactenum}[i)]
\item The rows associated to the players in any authorized set are linearly
independent.
\item The rank of $M$ is given by $1+r_1+\cdots+r_k$.
\item 
The rows associated with any given participant are linearly independent.
\item For each column $2\leq j\leq c+1$ in $M$, there are exactly two rows in $M$
which have nonzero support on $j$ and they both belong to the same authorized set.
\end{compactenum}
For a vector $v=(v_1,\ldots, v_n)\in \F_q^n$, we refer to the support of $v$
as the set of coordinates for which $v_i\neq 0$.

In the quantum secret sharing scheme due to Smith \cite{smith00} 
the encoded state is given by 
\ben
\ket{s} \mapsto \frac{1}{\sqrt{q^c}}\sum_{u\in \F_q^{c+1}:u_1=s} \ket{u M^t},\label{eq:mspQss}
\een
where we assume that the $u=(s,s_1\ldots,s_c) \in \F_q^{c+1}$. Please note that we do not
need to use any additional random variables in the quantum secret sharing schemes
as opposed to the classical schemes. We also note that the construction in \cite{smith00}
will hold for span matrices that are not in normal form; however, the associated access 
structure must still be self-dual.
The reader might find it helpful to refer to a small example given in the appendix.

We denote the  von Neumann entropy of a quantum state with the density matrix $\rho$ 
by $\s(\rho)=  -\Tr (\rho \log_2 \rho)$.
The following relation on the entropies of the authorized and unauthorized sets of
a quantum secret sharing scheme realized by monotone span programs was shown in 
\cite[Lemma~17]{rietjens05}.
\begin{lemma}[\cite{rietjens05}]\label{lm:entSets}
Suppose $\mc{M} =(\F_q, M, \psi )$  is a monotone span program realizing a quantum access structure 
$\Gamma$. Let $A \in \Gamma$ and $B=P\setminus A$. Let $\rk(M_A)=a$, $\rk(M_B)=b$ and
$\rk(M)=m$ and $\s(S)$ the entropy of the secret $S$.
 Then we have
\ben
\s(A) & = & (a+b-m)\log_2q + \s(S) \label{eq:entAuthSet}\\
\s(B) & = & (a+b-m)\log_2q \label{eq:entunAuthSet}
\een
\end{lemma}

\section{Entropic Inequalities}

In this section we  prove our main result for monotonicity of entropy for the 
unauthorized sets of a quantum secret sharing scheme realized via the
construction given in the previous section. We need the following lemma.

\begin{lemma}\label{lm:submatIndep}
Suppose that $\mc{M}=(\F_q,M,\psi)$ is a monotone span program in normal form 
computing the
access structure $\Gamma$. Let  $A\cup B \cup \{p \}=P$ be a partition of $P$ where $B\in 
\Gamma$. Let $A_i\in \Gamma_{\min}$ be a minimal authorized set that contains $p$.
Denote by $M_{\{A_i \}}$ the submatrix of $M$ with $|A_i|$ rows that correspond to
the players in $A_i$. Then one of the following holds:
\begin{compactenum}[i)]
\item If $A_i \subseteq \bar{A}$, then any row in $M_{\{A_i \}}$ is linearly dependent in 
$M_{\bar{A}}$ and linearly independent in $M_{A'}$, where $A'=A\cup \{ p \}$.
\item If $A_i \not\subseteq \bar{A}$, then any row in $M_{\{A_i \}}$ is linearly independent in 
$M_{\bar{A}}$ and linearly independent in $M_{A'}$.
\end{compactenum}
\end{lemma}
\begin{proof}
Since $A_i$ is authorized, $\varepsilon_1$ is in $\im (M_{A_i}^t)$ and
there exists some linear combination of the rows of $M_{A_i}$ such that
\be 
\sum_{j=1}^{|A_i|} R_j =\varepsilon_1,
\ee
where $R_j$ are the rows of $M_{A_i}$. Without loss of generality let $R_1$ be the
row associated with $p$ in $M_{\{A_i \}}$. 
Since $\bar{A}\setminus \{ p\} ={B}$ is still an authorized set, $\varepsilon_1$ is also in 
$\im(M_{{B}}^t)$ and there exists a linear combination of rows in ${B}$ such that
\be 
\sum_{j=1}^{|M_{B}|} \beta_j Q_j =\varepsilon_1 =  R_1 +\sum_{j=2}^{|A_i|} R_j,
\ee
where $Q_j$ are the rows of $M_{{B}}$ and $|M_B|$ denotes the number of rows of $M_B$. Then clearly, $R_1$
can be expressed as a linear combination of the rows in $M_{{B}}$. 

Next, we show that $R_1$ is independent in $M_{A'}$.
Assume that on the contrary that $R_1$
can also be expressed as a linear combination of the rows in $M_{A'}$ (excluding $R_1$). 
Now observe that the rows corresponding to $A_i$ are of two types: either they have a 
single non-zero element or they have $1+|A_i|$ nonzero elements. 
Permuting if necessary we can assume that $R_1$ is of the form $f=(0,1,0,\ldots, 0) \in \F_q^{c+1}$ or $g=(1,\mbf{-1},0,\ldots,0) \in \F_q^{c+1}$.  A necessary condition for 
$R_1$ be a linear combination of rows of $M_{A'}$ is that the support of these rows must 
contain the support of $R_1$. Keeping in mind that $A_i\setminus \{ p\} \subseteq B =\bar{A'}$, we
infer that if  $R_1=f$, then no row in $M_{A'}$ has overlap with the
support of $R_1$. If $R_1=g$, then 
some rows of $M_{A'}$ can have support in the first coordinate but not in the next 
$|A_i|-1$ coordinates, where $R_1$ is nonzero. Therefore it must be the case that $R_1$ is independent in 
$M_{A'}$. This proves the first part of the lemma.

Now let us consider the case when $A_i \not\subseteq \bar{A}$.  If $R_1=g$, then whether
it is in $M_{\bar{A}}$ or $M_{A'}$, the rows in $\bar{A}$ do not contain the support of
$R_1$ and similarly the rows in $M_{A'}$ do not contain the support of $R_1$. Hence it is
independent in both $M_{\bar{A}}$ and $M_{A'}$.  If $R_1=f$, then $g$ corresponds to some
other participant $p'$ who is in either $\bar{A}$ or $A'$ but not both as 
$A'\cap \bar{A} =\{ p\}$. Thus $g$ is in one of the  rows of $M_{A'}$ or $M_{\bar{A}}$ but not 
both. Consequently, $R_1$ is independent in one of $M_{A'}$ or $M_{\bar{A}}$.

Suppose that both $R_1$ and $g$ are in the same set $S$, where $S$ is either $A'$ or 
$\bar{A}$. Note that 
$A_i\not\subseteq \bar{A}$ and neither is $A_i\not\subseteq A'$ because 
$A'\not\in \Gamma$. It follows that $M_{\{ A_i\}}$ must contain other rows and the support 
of $g$ extends beyond the support of $R_1$ and the first coordinate. 
Then since $g$ is the only element whose support overlaps with $R_1$, any linear 
combination of rows that generates $R_1$ must include $g$. We could rewrite this linear
combination to express $g$ as a linear combination of the elements of $S$. But we have 
already seen that $g$ is linearly independent in both the sets $A'$ and $\bar{A}$.
Therefore, it follows that such a 
combination does not exist and $R_1$ is independent as well. This proves the second part
of the lemma.
\end{proof}

With this preparation we  are now ready to prove our central result.
\begin{theorem}\label{th:monotoneMSP}
Suppose that an  access structure $\Gamma$ is realized using the normal form 
monotone span construction. 
Let $A\subseteq B \subseteq P $.Then 
\ben
\ba{cl}\s(A) \leq \s(B) & \text{ if } A, B\not\in \Gamma \label{eq:unAuthEntropy}\\
\s(A) \geq \s(B)& \text{ if } A, B\in \Gamma \label{eq:authEntropy}
\ea \label{eq:entropyMonotone}
\een
\end{theorem}
\begin{proof}
Let us first show this result assuming that $\Gamma $ is a self-dual access structure. 
The proof relies on Lemma~\ref{lm:entSets}. 
Without loss of generality we can assume that $|B\setminus A | = 1$ in other words, $B\setminus A = \{ p\}$. 
Let 
\be
A^p =\{A_i \in \Gamma_{\min}: A_i \subseteq \bar{A} \mbox{ and } p\in A_i \}\\
\bar{A}^{p} = \{A_i \in \Gamma_{\min} : A_i \not\subseteq \bar{A} \mbox{ and }  p\in A_i\}
\ee

Suppose that  $A^p\neq \emptyset$. 
Then there exists some $A_i\in \Gamma_{\min}$ such that $p\in A_i \subseteq \bar{A}$. 
Consider the rows associated to this set in $M$, i.e.
$M_{A_i}$. By Lemma~\ref{lm:submatIndep}, this row is dependent in $M_{\bar{A}}$
and therefore removing it will not change the rank of the resulting submatrix. On the 
otherhand, by the same lemma we know that this row is independent in $M_{A'}$, where
$A'=A\cup \{ p\}$, therefore the rank of the matrix obtained by adding this
row to $M_A$ is greater by one.

Repeating
this process for all the $A_i$ in $A_p$, we obtain a matrix with fewer rows than $M_{\bar{A}}$ but having the same rank. On the other hand, the 
rank of the matrix obtained by  adding these rows to $M_{A}$ increases for each element in $A^p$.

Now consider an authorized set $A_i\in \bar{A}^p$. Since $A_i$ is not a subset of $\bar{A}$,
it follows that some participant of $A_i$ must be in $A$. By Lemma~\ref{lm:submatIndep},
the row associated with $p$ in $M_{\{A_i \}}$ is independent in $M_{\bar{A}}$ as well
as $M_{A'}$. Therefore, the rank of the submatrix obtained by removing this row
from $M_{\bar{A}}$ diminishes by one while the submatrix obtained by adding this row to 
$M_{A}$ increases by one. Once again repeating this process for all the 
$A_i \in \bar{A}^p$ we see the the rank of $M_{A'}$ increases by $|\bar{A}^p|$, while the
rank of the submatrix obtained by removing all the rows associated with $p$ in $\bar{A}^p$
reduces its rank by $|\bar{A}^p|$.

Therefore adjoining to $M_A$ all the rows associated with $p$ gives us $M_{\bar{B}}$ and 
while, removing them from  $M_{\bar{A}}$ gives $M_B$. From the preceding discussion we see that $\rk(M_{\bar{B}})=\rk(M_{\bar{A}} ) - |\bar{A}^p|$ and 
$\rk(M_{B}) =\rk(M_{A}) + |A^p| + | \bar{A}^p|$. By Lemma~\ref{lm:entSets}, the entropy of $B$ is given by 
\be
\frac{\s(B)}{\log_2 q}& =& \rk(M_B)+ \rk(M_{\bar{B}})-\rk(M) \\
&=&  \rk(M_{A}) + |A^p| + | \bar{A}^p| +\rk(M_{\bar{A}}) - |\bar{A}^p|\\&&-\rk(M)\\
&=& \rk(M_{A})+\rk(M_{\bar{A}})    -\rk(M)+ |A^p|\\
&=&\frac{\s(A)}{\log_2 q}+|A^p| \geq \frac{\s(A)}{\log_2 q}
\ee
If $|B\setminus A | >1$, we can inductively apply this argument to 
every consecutive  pair of sets in the following chain 
\be
A= B_k\subset B_{k-1}\subset \cdots \subset B_1\subset B_0=B,
\ee
where $|B_{i+1}\setminus B_i |=1$.
Applying to each adjacent pair in the above chain gives us
\be
\s(A) \leq \s(B_{k-1}) \leq \cdots \leq \s(B_1) \leq \s(B).
\ee
This proves the theorem for the case when the access structure is self-dual and
$A, B\not\in \Gamma$.

If  $A \subseteq B\in \Gamma$, 
 then we note that both $\bar{B} \subseteq \bar{A} \not\in \Gamma$ and we must have $S(\bar{B}) \leq S(\bar{A})$. But we also know that for a self-dual access structure $\s(\bar{A}) = \s(A)-\s(S)$ if $A$ is authorized \cite{rietjens05,imai04}. Therefore,
$\s(B)-\s(S) \leq \s(A)-\s(S)$ and this proves the theorem when $\Gamma$ is self-dual.

Now suppose that $\Gamma$ is not a self-dual access structure, then we can purify it to 
get a self-dual access structure $\overline{\Gamma}$ for which 
equation~\eqref{eq:entropyMonotone} holds. Recall that the 
authorized (unauthorized) sets of $\Gamma$ are  also authorized (unauthorized) sets of 
$\overline{\Gamma}$ and the associated shares are obtained by tracing out the additional
participant used for purification, see \cite{gottesman00} for details about purification. 
Therefore, the result holds for any quantum access structure implemented via the
normal form  monotone span construction.
\end{proof}

\begin{corollary}\label{co:maximalEntropy}For an access structure realized via the
normal form monotone span construction, the following relations hold:
Among the authorized sets the minimal authorized sets of an access structure have maximal 
entropy and among the unauthorized sets the maximal unauthorized sets have maximal 
entropy.
\end{corollary}
Please note the the above corollary does not imply that all minimal (maximal) authorized
(unauthorized) sets have the same entropy.
Further, along with Lemma~\ref{lm:entSets} it implies that if we consider a minimal authorized set, the entropy of the sets obtained by either
adding participants or removing participants from the minimal authorized set
will be lower.

\section{Discussion}
An obvious question is if these results can be extended to all quantum secret sharing 
states arising via monotone span programs and more generally, to all secret sharing 
states.  While this result might be extended to
other classes of secret sharing states,  it does not seem to generalize
for arbitrary quantum secret sharing states.
Nonetheless, these results  could be prove to be useful and provide additional
interesting insights into quantum secret sharing states 
in that by partly recovering the monotonicity for the von Neumann entropy we may be
able to prove new constrained  inequalities for the von Neumann entropy.

\section*{Acknowledgment}

This research is sponsored by CIFAR, MITACS and NSERC. 
I would like to thank Robert Raussendorf and Daniel Gottesman  for helpful discussions.

{
\appendix
\section{An Example}
We provide a small example to illustrate the details of the construction of quantum 
secret sharing schemes from monotone span programs.
Consider the minimal access structure  
$\Gamma_{\min} = \left\{\{1,2 \}, \{2,3 \},\{3,1 \}\right\}$. The span matrix $M$ for this
access structure is  given by 
\be
M = \left[\ba{rrrr} 
0&1 &0&0\\
1&-1&0&0\\
0&0&1&0\\
1&0&-1&0 \\
0&0&0&1\\
1&0&0&-1
\ea\right].
\ee
The function $\psi : \{1,2,3,4,5,6 \} \rightarrow P$, takes the values $\psi(1)=1$, $\psi(2)=2$, $\psi(3)=1$, $\psi(4)=3$, $\psi(5)=3$ and $\psi(6)=1$.
Following equation~\eqref{eq:mspQss}, we find that the encoding (up to
normalization) for the quantum secret 
sharing scheme is given by 
\be
\ket{0} \mapsto \ket{000000}+\ket{110000}+\ket{001100}+\ket{000011}\\
+\ket{111100}+\ket{110011}+\ket{001111}+\ket{111111}\\
\ket{1} \mapsto \ket{101010}+\ket{011010}+\ket{100110}+\ket{101001}\\
+\ket{011110}+\ket{011001}+\ket{100101}+\ket{010101}
\ee
}

Then by Theorem~\ref{th:monotoneMSP} and Lemma~\ref{lm:entSets} we compute the
entropy for the unauthorized sets $\{\{\emptyset\}, \{1 \} \}$ is $0, \log_2 q$
respetively, while for the authorized sets $\{ \{1,2 \}, \{1,2,3 \} \}$ it 
is $\s(S)+\log_2q$
and $\s(S)$, where $\s(S)$ is the entropy of the secret.

In general, for self-dual access structures,
if we start with the empty set and keep adding participants the entropy first increases
until it becomes a minimal authorized set and then starts decreasing until it reaches 
$\s(S)$, giving a ``tent-like'' characteristic. More precisely consider the following chain of sets 
$\emptyset=B_0 \subsetneq B_1\subsetneq \cdots  \subsetneq B_{n-1}\subsetneq B_n=P$, 
such that $|B_i\setminus B_{i-1}|=1$. Then only one of these sets is a minimal authorized
set, say $B_j$. If we now plot the entropy 
of these subsets we typically get a plot similar to the figure shown below, with 
 $\s(B_0)=0$ and $\s(B_n)=\s(S)$ and the 
entropy peaking at the minimal authorized set $B_j$.
(Please note that the figure below is only representative and does not correspond to any
quantum access structure. For simplicity,
we assume that the secret is a completely mixed state; thus 
$\s(S)=\log_2 q$.)

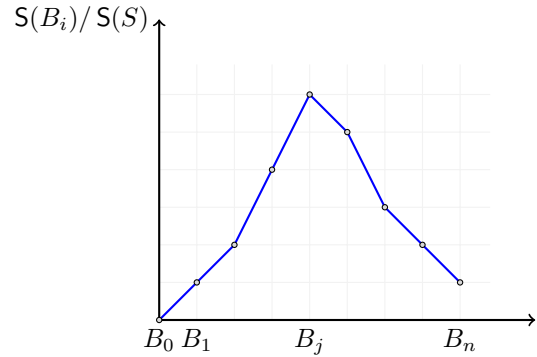
\begin{tikzpicture}[domain=0:2] 
\draw[thin,color=gray!10!white,step=.5cm] (0,0) grid (4.4,3.4); 
\draw[thick,color=black,->] (0,0) -- (5,0);  
\draw[thick, color=black,->] (0,0) -- (0,4) node[left] {$\s(B_i)/\s(S)$};
\draw[thick,color=blue] (0,0) -- (0.5,0.5) --plot coordinates {(1,1) (1.5,2)
(2, 3) (2.5,2.5) (3,1.5) (3.5,1) (4,0.5)};
\draw[color=black] (0.0,0) node[below] {$B_0$};
\draw[color=black] (0.5,0) node[below] {$B_1$};
\draw[color=black] (2.0,0) node[below] {$B_{j}$};
\draw[color=black] (4,0) node[below] {$B_n$};
\draw[fill=gray!30!white] (0,0) circle (1pt );
\draw[fill=gray!30!white] (0.5,0.5) circle (1pt );
\draw[fill=gray!30!white] (1,1) circle (1pt );
\draw[fill=gray!30!white] (1.5,2) circle (1pt );
\draw[fill=gray!30!white] (2,3) circle (1pt );
\draw[fill=gray!30!white] (2.5,2.5) circle (1pt );
\draw[fill=gray!30!white] (3,1.5) circle (1pt );
\draw[fill=gray!30!white] (3.5,1) circle (1pt );
\draw[fill=gray!30!white] (4,0.5) circle (1pt );
\end{tikzpicture}

The mixed-state schemes, i.e. those realizing non-self-dual access structures,  also show a similar but not exactly the same ``tent-like'' behavior in that $\s(B_n) \geq \s(S)$. 

%

\end{document}